\documentclass{article}

\usepackage{lipsum}
\usepackage{amsfonts}
\usepackage{graphicx}
\usepackage{epstopdf}
\usepackage{algorithmic}
\usepackage{geometry}
\usepackage{amsmath}
\usepackage{amsfonts,amsmath,amsthm,amssymb}
\usepackage{mathtools}
\usepackage[numbers,sort,compress]{natbib}
\RequirePackage[colorlinks,citecolor=blue,urlcolor=blue]{hyperref}

\ifpdf
  \DeclareGraphicsExtensions{.eps,.pdf,.png,.jpg}
\else
  \DeclareGraphicsExtensions{.eps}
\fi

\newcommand{\R}{\mathbb{R}}

\newcommand{\E}{\mathbb{E}}
\newcommand{\Acal}{{\mathcal A}}
\newcommand{\Fcal}{{\mathcal F}}
\newcommand{\Rcal}{{\mathcal R}}

\theoremstyle{plain}
\newtheorem{lem}{Lemma}
\newtheorem{prop}[lem]{Proposition}
\newtheorem{assum}[lem]{Assumption}
\newtheorem{eg}[lem]{Example}
\newtheorem{theorem}[lem]{Theorem}
\newtheorem{remark}[lem]{Remark}

\numberwithin{equation}{section}
\numberwithin{lem}{section}

\makeatletter
\renewenvironment{proof}[1][\proofname] {\par\pushQED{\qed}\normalfont\topsep6\p@\@plus6\p@\relax\trivlist\item[\hskip\labelsep\bfseries#1\@addpunct{.}]\ignorespaces}{\popQED\endtrivlist\@endpefalse}
\makeatother

\usepackage{enumitem}

\usepackage{enumitem}
\setlist[enumerate]{leftmargin=.5in}
\setlist[itemize]{leftmargin=.5in}

\title{Consumption-Investment Problem in Rank-Based Models
}

\author{David Itkin\footnote{Department of Statistics, London School of Economics and Political Science, \href{mailto:d.itkin@lse.ac.uk}{d.itkin@lse.ac.uk}.}}

\usepackage{amsopn}
\DeclareMathOperator{\diag}{diag}

\begin{document}

\maketitle

\begin{abstract}
We study a consumption-investment problem in a multi-asset market where the returns follow a generic rank-based model. 
Our main results are an HJB equation with Neumann boundary conditions for the value function, together with a corresponding verification theorem.
The control problem is nonstandard due to the discontinuous nature of the coefficients in rank-based models, requiring a bespoke approach of independent mathematical interest. The special case of first-order models, prescribing constant drift and diffusion coefficients for the ranked returns, admits explicit solutions when the investor (a) is unconstrained, (b) abides by open market constraints or (c) is fully invested in the market. The explicit optimal strategies in all cases are related to the celebrated solution to Merton's problem, despite case (b) being intractable in the classic Merton setting. 
\end{abstract}



\section{Introduction}
We study a consumption-investment problem in a high-dimensional equity market with a risk-free asset and $d$ equities $X=(X_1,\dots,X_d)$ available for investment. The problem is characterized by the value function
\begin{equation} \label{eqn:value_func}
v(t,x,w) = \sup_{(\pi,c) \in \Acal_T(t,x,w)} \E\bigg[\int_t^T e^{-\beta s}U_1\big(c(s)\big)ds + U_2\big(V^w_{\pi,c}(T)\big) \, \bigg |\,  X(t) = x \bigg], 
\end{equation}
which seeks to maximize an investor's utility from consumption and terminal wealth (all of the ingredients in \eqref{eqn:value_func} are precisely defined in Section~\ref{sec:consumption_investment}). Problems of this type have a rich history going back to Merton \cite{Merton1969} who solved the problem explicitly with power utility when the asset process $X$ is a geometric Brownian motion (GBM). Numerous extensions studying portfolio constraints, transaction costs 
and incomplete markets, 
among many others, have been considered throughout the years \cite{cvitanic1992convex,davis1990portfolio,karatzas1987optimal, kramkov1999asymptotic, shreve1994optimal,zariphopoulou2001solution}.

Here, we revisit this problem when the stock returns follow a \emph{rank-based} model. These are reduced-form models where each asset's return at a given time depends on the rank in the market the company occupies at that time (see \eqref{eqn:X_SDE} below). Rank-based models have been shown to reproduce certain features of large equity markets that standard models are unable to capture, such as the empirical stability of the capital distribution curve over long periods of time \cite{fernholz2002stochastic,itkin2024calibrated}. The class of first-order models, in particular, has seen a flurry of research activity in recent years studying various properties, from particle collisions 
to mean-field scaling limits as the dimension of the system tends to infinity 
\cite{dembo2016large, ichiba2013strong,ichiba2013convergence, sarantsev2015triple}. Moreover, the rank-based drift parameters can be efficiently estimated through a so-called collision estimator (see \cite[Chapter~5.4]{fernholz2002stochastic}) providing a path to drift parameter calibration, which is a notoriously difficult problem for standard name-based models. On the other hand, rank-based models inherently lead to discontinuous coefficients for the returns process, which creates difficulties in applying the standard stochastic optimal control machinery to characterize the value function \eqref{eqn:value_func}. The exception is the case of log utility, $U_1(x) = U_2(x) = \log(x)$, where a dynamic programming approach can be avoided as the myopic nature of the logarithmic criterion allows one to characterize the problem via a pointwise maximization problem (see \cite{karatzas2007numeraire,karatzas2021portfolio}). This has the advantage of allowing for certain pathwise, rather than in expectation, versions of the criterion, but also has the drawback that other investor preferences cannot be incorporated.

Our main result overcomes this difficulty by deriving the Hamilton--Jacobi--Bellman (HJB) equation  \eqref{eqn:HJB_ranked} that the value function satisfies, together with the appropriate Neumann boundary conditions \eqref{eqn:neumann_condition} and proving a corresponding verification result, Theorem~\ref{thm:verification}. The theory of HJBs with Neumann boundary conditions, as well as its connections to reflected diffusions, has a rich literature \cite{dupuis1991oblique,ren2019probabilistic} to which our results contribute by establishing a new probabilistic representation for the solution to the HJB equation (see Theorem~\ref{thm:verification}). When $X$ follows the \emph{first-order model} of \cite[Chapter~5.5]{fernholz2002stochastic}, prescribing constant drift and diffusion coefficients for the ranked returns, it turns out that an explicit solution exists. The optimal rule involves (i) the celebrated Merton fraction
\[\widetilde \pi^* = \tfrac{1}{\gamma}\widetilde a^{-1}(\widetilde \mu-r{\bf 1}_d),\]
where $\widetilde a$ is the instantaneous covariance matrix of the ranked returns, $\widetilde \mu$ is the drift of the ranked returns, $r$ is the risk-free rate and $\gamma$ is the investor's risk-aversion parameter, and (ii) the same optimal feedback-form consumption rule $c^*$ as in Merton's problem. The key difference is that the fraction $\widetilde \pi^*$ specifies the optimal proportion of wealth the investor should hold in the assets according to the rank they occupy, rather than by company name. 

Additionally, in Section~\ref{sec:open_markets}, we study \emph{open market} constraints, which have recently gained attention in the literature \cite{itkin2024open,karatzas2021open}. These constraints only allow investment in certain ranks and, as such, can serve as a proxy for investors who restrict their investment universe to companies within a certain size range (e.g., large-cap equities, mid-cap equities). Remarkably, in the case of first-order models we obtain entirely explicit optimal solutions, whereas the corresponding open market constrained problem in the standard GBM setting is intractable. Section~\ref{sec:fully_invested} derives further explicit optimal strategies when the investor is additionally required to be fully invested in the market.

On the mathematical side, we characterize the value function for this problem by relating it to an HJB equation with Neumann boundary conditions. Equations of this type on general domains were recently studied in \cite{ren2019probabilistic} and shown to characterize the value function with a \emph{reflected} diffusion, corresponding here to the ranked capitalizations $X_{()}$, as one of the state variables. In contrast, the named capitalizations $X$ are state variables in \eqref{eqn:value_func}. As such, the admissible controls in the ranked problem are predictable with respect to the filtration generated by $X_{()}$ and its driving noise $B$, while in our problem the filtration is larger, incorporating information about $X$ and its driving noise $W$. Our main result, Theorem~\ref{thm:verification}, establishes that the value functions for these two problems coincide despite
the fact that the optimal allocations $\pi^*$ for \eqref{eqn:value_func} may not be adapted to the filtration generated by $(X_{()},B)$. 

The paper is organized as follows. Section~\ref{sec:consumption_investment} introduces the model and the consumption-investment problem. In Section~\ref{sec:dynamic_programming} we study the associated HJB equation. It is heuristically derived in Section~\ref{sec:heuristic} and the verification theorem is established in Section~\ref{sec:verification}. Finally, in Section~\ref{sec:first_order} we consider first-order models with power utility preferences, explicitly solving the unconstrained problem in Section~\ref{sec:unconstrained}, the open-market-constrained problem in Section~\ref{sec:open_markets} and the fully invested problem in Section~\ref{sec:fully_invested}.

\paragraph{Notation} For $d \geq 2$ and a vector $x \in (0,\infty)^d$, we write $x_{()}$ for its decreasing order statistics, which is the vector of components of $x$ arranged in decreasing order satisfying $x_{(1)} \geq x_{(2)} \geq \dots \geq x_{(d)}$. The set of all ordered vectors is defined as
$(0,\infty)^d_{\geq} = \{y \in (0,\infty)^d: y_1 \geq y_2 \geq \dots \geq  y_d\}$.
The subset of $(0,\infty)^d_{\geq}$ consisting of all points where no coordinates are equal will be denoted by $(0,\infty)^d_{>}$ and its complement by $(0,\infty)^d_{=} = (0,\infty)^d_{\geq} \setminus (0,\infty)^d_>$. 

For every $i \in \{1,\dots,d\}$, we define the rank identifying function $\Rcal_i:(0,\infty)^d \to \{1,\dots,d\}$ via $\Rcal_i(x) = k$ where $k$ is chosen so that $x_i = x_{(k)}$ and, among indices tied with $i$, ties are broken using lexicographical ordering. That is, we take $\Rcal_i(x) < \Rcal_j(x)$ if $i < j$ and $x_i = x_j$, though the precise tie-breaking rule does not affect the results in this paper.
\section{The consumption-investment problem} \label{sec:consumption_investment}
\subsection{Financial market} We consider a financial market consisting of a risk-free asset $dX_0(t) = rX_0(t)dt$
with risk-free rate $r \in \R$ and $d$ risky assets with market capitalizations $X = (X_1,\dots,X_d)$. We take a rank-based model for the risky assets,
\begin{equation} \label{eqn:X_SDE}
    \frac{dX_i(s)}{X_i(s)}= \mu_i(s,X(s))ds +  \sum_{j=1}^d\sigma_{ij}(s,X(s))dW_j(s), \quad X_i(t) = x_i; \qquad i=1,\dots,d
\end{equation}
for any initial time $ t \geq 0$, all $s \ge t$ and for initial value $x \in (0,\infty)^d$,
where 
\[\mu_i(s,x) = \textstyle \sum_{k=1}^d \widetilde \mu_k(s,x_{()})1_{\{\Rcal_i(x) = k\}}, \qquad \sigma_{ij}(s,x) = \sum_{k,\ell=1}^d \widetilde \sigma_{k\ell}(s,x_{()})1_{\{\Rcal_i(x) = k, \Rcal_j(x) = \ell\}}\]
for some functions $\widetilde \mu_k,\widetilde \sigma_{k\ell}:[0,\infty)\times (0,\infty)^d_{\ge} \to \R$
and standard Brownian motion $W$. In this model an asset's dynamics only depend on which rank that asset occupies at any time. Even under strong Lipschitz assumptions on $\widetilde\mu$ and $\widetilde\sigma$, the coefficients $\mu$ and $\sigma$ are typically discontinuous and fail to be weakly differentiable, so we cannot expect strong solutions or uniqueness in law for \eqref{eqn:X_SDE}. We therefore work with weak solutions, under the following minimal assumptions on the coefficients.
\begin{assum} \label{assum:coefficients}
$\widetilde \mu$ and $\widetilde \sigma$ are bounded and $\widetilde a = \widetilde \sigma \widetilde \sigma^\top$ is uniformly elliptic. 
\end{assum}

\begin{remark}
    We stress that for solutions to exist for the HJB equation \eqref{eqn:HJB_ranked} derived in the sequel, one will typically need to assume more on the coefficients; in particular, one expects to need at least continuity of $\widetilde \sigma$. Since our verification theorem, Theorem~\ref{thm:verification}, assumes the existence of a sufficiently regular solution to the HJB, we do not need to explicitly impose stronger regularity conditions than  Assumption~\ref{assum:coefficients} prescribes. 
\end{remark}
We now assert existence of our asset price process $X$.

\begin{prop} \label{prop:SDE}
    Under Assumption~\ref{assum:coefficients},  there exists a weak solution to \eqref{eqn:X_SDE} for every initial $t \geq 0$ and $x \in (0,\infty)^d$.
\end{prop}
\begin{proof}
It suffices to establish weak existence on $\R^d$ for the process $Z = (\log X_1,\dots, \log X_d)$, which satisfies the SDE
    \[dZ(s) = \Big(\mu(s,\exp Z(s)) - \tfrac{1}{2}\diag\big(a(s,\exp Z(s))\big)\Big)ds + \sigma(s,\exp Z(s))dW(s),\]
    where $a = \sigma \sigma^\top$ and $\exp(z) = (\exp(z_1),\dots,\exp(z_d))$. By Assumption~\ref{assum:coefficients}, the drift and volatility coefficients of $Z$ are measurable and bounded, and the diffusion matrix is uniformly elliptic. Hence, \cite[Theorem~2.6.1]{krylov1980controlled} guarantees the existence of a weak solution.
\end{proof}

    To establish our main result, Theorem~\ref{thm:verification}, we will need certain $L^p$ bounds for the process $X$. To this end, we establish the following lemma.
\begin{lem} \label{lem:Gronwall}
    Let $Z = (Z_1,\dots,Z_d)$ be a stochastic process satisfying 
    \begin{equation}
        \label{eq:Gronwall} \frac{dZ_i(s)}{Z_i(s)} =  \alpha_i(s) ds + \sum_{j=1}^d \nu_{ij}(s) dW_j(s); \qquad Z_i(t) = z_i, \quad i=1,\dots,d, 
    \end{equation}
    for $s \ge t$, some initial value $z \in (0,\infty)^d$ and bounded predictable processes $\alpha(\cdot)$ and $\nu(\cdot)$. Then, we have the bound
    \begin{equation} \label{eqn:Lp_RSDE}
         \E\Big[\sup_{t \leq s\leq T} \|Z(s)\|^p\, \Big| \, Z(t) = z\Big] \leq C(1+\|z\|^p) 
    \end{equation}
    for any $T \geq t$, $p \geq 1$ and some constant $C>0$.
\end{lem}
\begin{proof} If $Z_i$ satisfies \eqref{eq:Gronwall} then it has the representation
\[Z_i(s) = z_i\exp\bigg( \int_t^s \alpha_i(u)du\bigg)\mathcal{E}\bigg(\sum_{j=1}^d\int_t^s \nu_{ij}(u)dW_j(u)\bigg),
\]
where $\mathcal{E}(\cdot)$ denotes the Dol\'eans--Dade exponential.
We now raise both sides to the power $p$, sum over $i=1,\dots,d$ and take pathwise supremum up to time $T$. Since $\nu$ is bounded, for each $p\ge 1$ the process $\mathcal E(p\sum_j\int_t^\cdot\nu_{ij}dW_j)$ is a true martingale.
Combined with boundedness of the drift integral and Doob's maximal inequality this yields \eqref{eqn:Lp_RSDE} for $p > 1$ and a constant $C >0$ that can be explicitly worked out if desired. The case $p = 1$ follows from the Cauchy--Schwarz inequality and the already-established $p=2$ case.
\end{proof}
Note that the relationship \eqref{eq:Gronwall} holds for the process $X$, so we can apply Lemma~\ref{lem:Gronwall} to $X$. Since $\|x\| = \|x_{()}\|$ for all $x \in (0,\infty)^d$, the bound \eqref{eqn:Lp_RSDE} also holds for the process $X_{()}$. 

\subsection{Ranked market capitalizations}
 We now study the ranked market capitalizations. By \cite[Theorem~2.5]{banner2008local}, the ranked process $Y = X_{()}$ satisfies
\begin{equation} \label{eqn:RSDE}
    dY_k(s) =Y_k(s)\Big(\widetilde \mu_k(s,Y(s))ds + {\textstyle \sum_{\ell=1}^d }\widetilde \sigma_{k\ell}(s,Y(s))dB_\ell(s)\Big) + d\Phi_k(s), \quad Y_k(t) = y_k,
\end{equation}
for $k=1,\dots,d$ and $s \geq t$, where $y_k = x_{(k)}$, $B_\ell(s) = \int_t^s\sum_{j=1}^d1_{\{\Rcal_j(X(u)) = \ell\}}dW_j(u)$ and the process $\Phi = (\Phi_1,\dots,\Phi_d)$ is an appropriate reflection term (see Assumption~\ref{ass:triple_collision} below and the subsequent discussion). Note that the semimartingale decomposition in \cite[Theorem~2.5]{banner2008local} is applicable to $X_{()}$ by the remark preceding their Theorem~2.5. By L\'evy's characterization of Brownian motion, $B$ is a standard $d$-dimensional Brownian motion. 
    Thus, \eqref{eqn:RSDE} is a reflected stochastic differential equation (RSDE) with normal reflection on the domain $(0,\infty)^d_{\geq}$. For each initial condition, we fix a weak solution $(X,W)$ to \eqref{eqn:X_SDE}, set $(Y,B)$ as above and write $\mathcal{F}^{X,W}, \mathcal{F}^{Y,B}$ and $\Fcal^B$ for the augmented filtrations generated by $(X,W)$, $(Y,B)$ and $B$ respectively. Since $(Y,B)$ is a functional of $(X,W)$, we have $\Fcal^{Y,B} \subseteq \Fcal^{X,W}$, and $W$ and $B$ are Brownian motions with respect to $\Fcal^{X,W}$ and $\Fcal^{Y,B}$ respectively. We do not assume \eqref{eqn:RSDE} is strongly well-posed, but note that when it is, the process $Y$ is adapted to $\Fcal^B$. 

Now, we obtain a simple characterization for the reflection term $\Phi$, which will be useful in the following sections. For any univariate continuous semimartingale $Z$ we write $L_Z$ for its local time process at zero, $L_Z(s) = |Z(s)| - |Z(t)| -\int_t^s \mathrm{sign}(Z(u))dZ(u)$. 
\begin{assum} \label{ass:triple_collision}
     We assume that $L_{X_{(k)} - X_{(\ell)}} = 0$ whenever $k - \ell > 1$.
\end{assum}
Under this assumption, we have the representation $\Phi_k(s) = \frac{1}{4}L_{Y_k-Y_{k+1}}(s) - \frac{1}{4}L_{Y_{k-1}-Y_k}(s)$ for $k =1,\dots,d$, with the convention $L_{Y_0-Y_1}=L_{Y_d-Y_{d+1}} = 0$, courtesy of  \cite[Theorem~2.5]{banner2008local}. Assumption~\ref{ass:triple_collision} is a weak condition which, for instance, holds in first-order models (see \cite[Lemma~1]{ichiba2011hybrid}) or when three or more components of $X$ never simultaneously collide.

\subsection{The stochastic optimal control problem}
When starting with initial wealth $w > 0$, using the trading strategy $\pi(s) = (\pi_1(s),\dots,\pi_d(s))$ and consuming at rate $c(s) \geq 0$ for $s \geq t$, the investor's wealth process $V_{\pi,c}^w(\cdot)$ evolves according to
\begin{align*}
    dV_{\pi,c}^w&(s)
     = V_{\pi,c}^w(s)\bigg(\sum_{i=1}^d \pi_i(s) \frac{dX_i(s)}{X_i(s)} + \pi_0(s)\frac{dX_0(s)}{X_0(s)}\bigg) - c(s)ds \\
    & =\bigg(V^w_{\pi,c}(s)\Big(\pi(s)^\top \big(\mu(s,X(s))-r{\bf 1}_d\big) + r\Big)- c(s)\bigg)ds + V^w_{\pi,c}(s)\pi(s)^\top \sigma(s,X(s))dW(s),
\end{align*}
where $\pi_0(s) = 1-\sum_{i=1}^d \pi_i(s)$ is the proportion of wealth invested in the risk-free asset and ${\bf 1}_d$ denotes the $d$-dimensional vector of all ones. Given a time horizon $T \in (0,\infty)$,
our goal is to solve the consumption-investment problem characterized by the value function  \eqref{eqn:value_func} given in the introduction, where $\beta > 0$ is the patience parameter, $\Acal_T(t,x,w)$ is the set of admissible controls specified in Section~\ref{sec:control_set} and $U_1,U_2:(0,\infty) \to (0,\infty)$ are utility functions, that we assume to be concave and increasing. Utility functions on the real line or an infinite time horizon can be handled similarly, but we do not pursue this here.

\begin{remark}
    Due to the potential nonuniqueness in law of solutions to \eqref{eqn:X_SDE}, it is not a priori clear that the value function $v$ depends only on $(t,x,w)$ and not on the particular solution to \eqref{eqn:X_SDE}. Part of our main result, Theorem~\ref{thm:verification}, establishes that $v$ is independent of the particular solution.
\end{remark}

\subsection{Control set} \label{sec:control_set}
Here we define the set of controls $\Acal_T(t,x,w)$ over which we optimize. We assume it is of the form
\begin{equation} \label{eqn:control_set}  
\Acal_T(t,x,w) = \{(\pi(s),c(s))_{s\in[t,T]}: V^w_{\pi,c}(s) > 0 \text{ \,and\, }(\pi(s),c(s)) \in A(X(s),V_{\pi,c}^w(s)) \,   \forall s \in [t,T]\}
\end{equation}
for some nonempty measurable subsets $A(x,w) \subset \R^d \times [0,\infty)$, which specify constraints on the $\Fcal^{X,W}$-predictable  control processes $(\pi(\cdot),c(\cdot))$. Next, we define the rank-based constraint sets
\[
     \widetilde A(x,w)  = \{(\widetilde \pi,c) \in \R^d \times [0,\infty): \widetilde \pi_k = {\textstyle\sum_{i=1}^d} \pi_i1_{\{\Rcal_i(x)=k\}} \,  \forall k \text{ and some } (\pi,c) \in A(x,w)\}.
\]
 If at some point in time the investor's portfolio weights are $\pi \in \R^d$, then $\widetilde \pi_k$ is the proportion of wealth the investor holds in the asset currently occupying rank $k$. 
 The main condition we will impose is permutation invariance of the rank-based constraint sets.
\begin{assum} \label{ass:constraint_set}
 $\widetilde A(x,w) = \widetilde A(x_{()},w)$ for every $(x,w) \in (0,\infty)^{d+1}$.
\end{assum}
We now present a few examples of common constraints that satisfy Assumption~\ref{ass:constraint_set}.
\begin{eg} \label{eg:constraints} \phantom{.}
\begin{enumerate}[noitemsep]
    \item \label{item:unconstrained} (Unconstrained).\ $A(x,w) = \R^d \times [0,\infty)$,
    \item  \label{item:long-only}  (Long-only).\ $A(x,w) = [0,\infty)^{d+1}$,
    \item \label{item:fully_invested} (Fully invested).\ $A(x,w) = \{\pi \in \R^d: \pi_1 + \dots + \pi_d  =1\}\times [0,\infty)$,
     \item  \label{item:open_market}  (Open market).\ $A(x,w) = \{\pi \in \R^d: \pi_i = 0 \text { if } \mathcal{R}_i(x)\not \in I\} \times [0,\infty)$ for some nonempty index set $I \subset \{1,\dots,d\}$.
 \end{enumerate}   
 \end{eg}
 That items \ref{item:unconstrained}--\ref{item:fully_invested} satisfy Assumption~\ref{ass:constraint_set} is clear, while for \ref{item:open_market} we note that $I$ is an index set that specifies which ranked assets are available for investment. The corresponding rank-based constraint set 
 $\widetilde A = \{(\widetilde \pi, c) \in \R^d \times [0,\infty): \widetilde \pi_k = 0 \text{ if } k \not \in I\}$ is state-independent and, hence, satisfies Assumption~\ref{ass:constraint_set}. An example of a constraint that does not satisfy Assumption~\ref{ass:constraint_set} is an asset-based long-only constraint, such as $A(x,w) = [0,\infty) \times \R^{d-1} \times [0,\infty)$ which prohibits short selling in asset one but allows it in the other assets.
 
\section{Dynamic programming approach} \label{sec:dynamic_programming} In this section, we will work with the domains 
\[  D_\geq = [0,T)\times (0,\infty)^d_{\geq}\times (0,\infty), \, \, D_> = [0,T)\times (0,\infty)^d_{>}\times (0,\infty), \, \, D_= = [0,T)\times (0,\infty)^d_{=}\times (0,\infty).\]
 When the right endpoint $\{T\}$ is included, we will append the corresponding set with a bar, such as $\overline D_{\geq}$.
For a function $\widetilde v:D_{\geq} \to \R$, we will write $\partial_t \widetilde v$ for the derivative in the first argument, $\nabla \widetilde v$ for the $(d+1)$-dimensional gradient in the last two arguments, $\nabla_y \widetilde v$ for the gradient in the second argument and $\partial_w \widetilde v$ for the derivative in the third argument. Similarly, $\nabla^2 \widetilde v$ will denote the $(d+1)\times (d+1)$ Hessian while  $\nabla_y^2 \widetilde v$, $\partial_w \nabla_y \widetilde v$ and $\partial_{ww}\widetilde v$ will denote its corresponding blocks. 

\subsection{Heuristic discussion} \label{sec:heuristic} 

The HJB associated to the value function is a second-order PDE whose coefficients $\mu(t,x)$ and $a(t,x) = (\sigma \sigma^\top)(t,x)$ are, in general, only measurable and bounded even if the coefficients $\widetilde \mu$ and $\widetilde \sigma$ are more regular. As such, the standard viscosity solution theory is not applicable.
To proceed, we make the permutation-invariant ansatz  \begin{equation}\label{eqn:v_ansatz}
    v(t,x,w) = \widetilde v (t,x_{()},w)
\end{equation} for some function $\widetilde v:D_{\geq}\to \R$. For any $k$, the map $x \mapsto x_{(k)}$ is Lipschitz and differentiable outside the Lebesgue null-set $\{x:x_i = x_j \text{ for some } i \ne j\}$ with $\partial_{x_i} x_{(k)}  = 1_{\{\Rcal_i(x) = k\}}$. As such, by combining these identities with the formal HJB for $v$, we expect $\widetilde v$ to satisfy 
\begin{equation} \label{eqn:HJB_ranked}
    \begin{cases}
          0 = (\partial_t + L^Y) \widetilde v(t,y,w) + \sup_{(\widetilde \pi,c) \in \widetilde A(y,w)}\widetilde H_{\widetilde \pi,c}\big(t,y,w,\nabla \widetilde v(t,y,w),\nabla^2 \widetilde v(t,y,w)\big), &  \hfill \text{in } D_>, \\
        \widetilde v(T,y,w) = U_2(w), & \hspace{-1.4cm}\text{in }(0,\infty)^d_{\geq} \times (0,\infty),
    \end{cases} 
\end{equation}
where $L^Y \widetilde v = \sum_{k=1}^dy_k\widetilde \mu_k \partial_{y_k} \widetilde v  + \frac{1}{2}\sum_{k,\ell=1}^dy_ky_\ell\widetilde a_{k\ell}\partial_{y_ky_\ell}\widetilde v$  and
\begin{align*}
    \widetilde H_{\widetilde \pi,c}(t,y,w,\xi,M) & =   w\xi_{d+1}\big(r+\widetilde \pi^\top(\widetilde \mu(t,y)-r{\bf 1}_d)\big) 
   + wM_{d+1,1:d}\diag(y) \widetilde a(t,y)\widetilde \pi \\
   & \qquad + \tfrac{w^2}{2}M_{d+1,d+1}\widetilde \pi^\top  \widetilde a(t,y)\widetilde \pi - \xi_{d+1}c + e^{-\beta t}U_1(c).
\end{align*}
The diffusion matrix $\widetilde a = \widetilde \sigma \widetilde \sigma^\top$ does not degenerate on the boundary $D_{=}$ so the equation \eqref{eqn:HJB_ranked} needs to be supplemented with boundary conditions. As $X_{()}$ satisfies an RSDE with normal reflection, we will impose Neumann boundary conditions on $\widetilde v$. Since higher order local times are not accumulated, courtesy of Assumption~\ref{ass:triple_collision}, we will only require the Neumann boundary conditions to hold on each face of the boundary, 
\begin{equation} \label{eqn:neumann_condition}
    \partial_{y_k} \widetilde v(t,y,w) - \partial_{y_{k+1}} \widetilde v(t,y,w) = 0, \quad \text{for } (t,y,w) \in D_{=} \cap \{y_k = y_{k+1}\}, \quad  k =1,\dots,d-1.
\end{equation}
The important role of the boundary condition \eqref{eqn:neumann_condition} will become apparent in the proof of the verification theorem in the next section.

\subsection{Verification theorem} \label{sec:verification}
HJB equations of the type \eqref{eqn:HJB_ranked} on fairly general domains with Neumann boundary conditions \eqref{eqn:neumann_condition} have recently been studied in \cite{ren2019probabilistic}. Under suitable assumptions on the domain and coefficients, the authors are able to establish the existence of a unique viscosity solution for HJB equations and boundary conditions of this type, characterizing it as the value function of a certain stochastic control problem. In our setting this value function is formally given by
\begin{equation} \label{eqn:value_func_ranked} 
\widetilde v(t,y,w) = \sup_{(\widetilde \pi,c) \in \widetilde \Acal_T(t,y,w)} \E\bigg[\int_t^T e^{-\beta s}U_1\big(c(s)\big)ds + U_2\big(\widetilde  V^w_{\widetilde \pi,c}(T)\big) \, \bigg|\, Y(t) = y\bigg],
\end{equation}
where $Y$ is as in \eqref{eqn:RSDE}, the process $\widetilde V^w_{\widetilde \pi,c}$ is given by
\[d\widetilde V^w_{\widetilde \pi,c}(s) = \bigg(\widetilde V^w_{\widetilde \pi,c}(s)\Big(\widetilde \pi(s)^\top \big(\widetilde \mu(s,Y(s))-r{\bf 1}_d\big) + r\Big)- c(s)\bigg)ds + \widetilde V^w_{\widetilde \pi,c}(s)\widetilde \pi(s)^\top \widetilde  \sigma(s,Y(s))dB(s),\] 
for $s \geq t$ with $\widetilde V^w_{\widetilde \pi,c}(t) = w$ and $\widetilde \Acal_T(t,y,w)$ denotes the set of admissible control processes. This set is akin to \eqref{eqn:control_set} with $\widetilde A(y,w)$ in place of $A(x,w)$ except that it is additionally restricted to processes $(\widetilde \pi(s),c(s))_{s\in [t,T]}$ which are predictable with respect to $\Fcal^{Y,B}$.

It is far from clear whether the value function $v$ satisfies the relationship \eqref{eqn:v_ansatz} with $\widetilde v$ solving \eqref{eqn:HJB_ranked}. Indeed, the representation \eqref{eqn:value_func_ranked} only allows for controls adapted to $\Fcal^{Y,B}$, while our control problem with value function $v$ allows for controls adapted to the larger filtration $\Fcal^{X,W}$. Nevertheless, we will show below that the relationship \eqref{eqn:v_ansatz} holds where $\widetilde v$ is a solution to \eqref{eqn:HJB_ranked} and \eqref{eqn:neumann_condition}. For clarity of exposition, and since it is sufficient for our main example in Section~\ref{sec:first_order}, we prove a verification theorem in the case of a classical solution to the HJB.
\begin{theorem}[Verification theorem] \label{thm:verification} Let Assumptions~\ref{assum:coefficients}, \ref{ass:triple_collision} and \ref{ass:constraint_set} hold.
    Suppose there exists a nonnegative function $\widetilde v \in C^{1,2}(D_{\geq})\cap C(\overline D_{\geq})$ satisfying \eqref{eqn:HJB_ranked}, \eqref{eqn:neumann_condition} as well as the polynomial growth condition
    \begin{equation} \label{eqn:v_poly_growth} 
    \widetilde v(t,y,w) \leq C(1+\|y\|^p+w^p) \quad \text{for all }(t,y,w) \in  \overline D_{\geq},
    \end{equation} for some $C > 0$ and $p \ge 1$. Suppose further that there exists a measurable maximizer
    \begin{equation}  \label{eqn:optimal_control_argmax}
     \big(\widetilde \pi^*(t,y,w),\widetilde c^*(t,y,w)\big) \in {\arg\max}_{(\widetilde \pi, c) \in \widetilde A(y,w)} \widetilde H_{ \widetilde \pi,c}\big(t,y,w,\nabla \widetilde v(t,y,w),\nabla^2 \widetilde v(t,y,w)\big)
    \end{equation}
    for every $(t,y,w) \in D_\geq$
    satisfying the uniform Lipschitz and linear growth condition
    \begin{equation} \label{eqn:Lipschitz_control}
      \begin{aligned}  
      &\|w_1\widetilde \pi^*(t,y,w_1) - w_2\widetilde \pi^*(t,y,w_2)\| + |\widetilde c^*(t,y,w_1)-\widetilde c^*(t,y,w_2)|  \leq L|w_1-w_2|, \\
       & \|w_1 \widetilde \pi^*(t,y,w_1)\| + |\widetilde c^*(t,y,w_1)| \leq Cw_1
      \end{aligned}
    \end{equation}
    for some $C,L > 0$ and all $t \in [0,T)$, $y \in (0,\infty)^d_{\geq}$ and $w_1,w_2 \in (0,\infty)$. Then, the value function \eqref{eqn:value_func} is given by \eqref{eqn:v_ansatz}, i.e.\ $v(t,x,w) = \widetilde v(t,x_{()},w)$, and an optimal strategy  $(\pi^*,c^*)$ is given in feedback form by
    \begin{align*} &\pi^*_i\big(s,X(s),V^w_{\pi^*,c^*}(s)\big)  = {\textstyle \sum_{k=1}^d} \widetilde \pi^*_k\big(s,X_{()}(s),V^w_{\pi^*,c^*}(s)\big)1_{\{\Rcal_i(X(s)) = k\}}, & i=1,\dots,d, \quad s \in [t,T),\\
        & c^*\big(s,X(s),V^w_{\pi^*,c^*}(s)\big) = \widetilde c^*\big(s,X_{()}(s),V^w_{\pi^*,c^*}(s)\big), & s\in [t,T).
    \end{align*}
\end{theorem}
\begin{remark}
    The processes $\widetilde \pi^*,\widetilde c^*$ and $c^*$ are $\Fcal^{Y,B}$-adapted, but the optimal control $\pi^*$ need not be. Additionally, note that the condition \eqref{eqn:Lipschitz_control} ensures that $\pi^*$ is bounded.  
\end{remark}

\begin{proof}
    We fix initial values $(t,x,w) \in [0,T) \times (0,\infty)^{d+1}$ and
    let $(\pi,c) \in \Acal_T(t,x,w) $ be arbitrary. In the course of this proof, to simplify the exposition, we will write $Z(s) = (X_{()}(s),V^w_{\pi,c}(s))$.  Applying It\^o's formula to $\widetilde v(s,Z(s))$, adding and subtracting $e^{-\beta s}U_1(c(s))ds$ and using the fact that $X_{()}$ satisfies \eqref{eqn:RSDE} gives
    \begin{equation}
    \begin{aligned}
        \label{eqn:ito_value_func} 
    d\widetilde v(s,Z(s)) = & \Big((\partial_t+L^Y)\widetilde v(s,Z(s))+ \widetilde H_{\widetilde \pi(s),
 c(s)}\big(s,Z(s),\nabla \widetilde v(s,Z(s)),\nabla^2 \widetilde v(s,Z(s))\big)\Big)ds \\
    & \quad  - e^{-\beta s}U_1\big(c(s)\big)ds+dM(s) + \nabla_y\widetilde v\big(s,Z(s)\big)^\top d\Phi(s), 
    \end{aligned}
    \end{equation} 
where $M$
is a local martingale and $\widetilde \pi_k(s) = \sum_{i=1}^d \pi_i(s)1_{\{\Rcal_i(X(s)) = k\}}$. Here we used the fact that $\pi(s)^\top \mu(s,X(s)) = \widetilde \pi(s)^\top \widetilde \mu(s,X_{()}(s))$
and $\pi(s)^\top a(s,X(s))\pi(s) = \widetilde \pi(s)^\top \widetilde a(s,X_{()}(s)) \widetilde \pi(s)$ to rewrite the Hamiltonian in terms of $(\widetilde \pi,c)$. 

By Assumption~\ref{ass:constraint_set} on the constraint set we have that $(\widetilde \pi(s),c(s)) \in \widetilde A(X_{()}(s),V^w_{\pi,c}(s))$ so we deduce from \eqref{eqn:HJB_ranked} that the first line on the right-hand side of \eqref{eqn:ito_value_func} is nonpositive. Additionally, the Neumann boundary condition \eqref{eqn:neumann_condition} together with Assumption~\ref{ass:triple_collision} ensures that the reflection terms vanish. As such, we obtain
\begin{equation} \label{eqn:value_func_bound} 
\widetilde v(u,Z(u)) \leq \widetilde v(t,x_{()},w) -\int_t^{u} e^{-\beta s}U_1\big(c(s)\big)ds + M(u) -M(t)
\end{equation}
for all $t \leq u \leq T$.
We now define $\tau_n = \inf\{ s\ge t: |M(s)| \geq n\}$, which specifies a localizing sequence of stopping times, and note that $M(\cdot \land \tau_n)$ is a bounded martingale for every $n$. Evaluating \eqref{eqn:value_func_bound} at $u=
T \land \tau_n$ and taking expectation yields
\begin{equation} \label{eq:verification_bound}
    \E\big[\widetilde v\big(T\land \tau_n,Z(T\land \tau_n)\big) \, \big| \, X(t) = x\big]  \leq \widetilde v(t,x_{()},w) - \E\bigg[\int_t^{T \land \tau_n} e^{-\beta s} U_1\big(c(s)\big)ds\, \bigg | \,  X(t) = x\bigg]. 
\end{equation} We obtain
    $\E[U_2(V^w_{\pi,c}(T))|X(t) = x]   \leq \widetilde v(t,x_{()},w) - \E[\int_t^T e^{-\beta s}U_1(c(s))ds|X(t)=x]$ by sending $ n \to \infty$, 
where we used Fatou's lemma and the terminal condition from \eqref{eqn:HJB_ranked} on the left-hand side and monotone convergence on the right-hand side. 
This establishes that
\begin{equation} \label{eqn:upper_bound}
v(t,x,w) = \sup_{(\pi,c) \in \Acal_T(t,x,w)} \E\bigg[\int_t^T e^{-\beta s}U_1\big(c(s)\big)ds + U_2\big(V^w_{\pi,c}(T)\big) \, \bigg|\, X(t) = x\bigg] \leq \widetilde v(t,x_{()},w).
\end{equation}
To obtain the reverse inequality we fix initial values $(t,x,w)$ and consider the feedback control $(\pi^*(s),c^*(s))_{s\in[t,T]}$. 
 The corresponding wealth process $V^* = V^w_{\pi^*,c^*}$ has dynamics
 \begin{equation} \label{eqn:V*}
\begin{aligned}
    dV^*(s) & = \bigg(V^*(s)\Big(r+\widetilde \pi^*(s,X_{()}(s),V^*(s))^\top\big(\widetilde \mu(s,X_{()}(s))-r{\bf 1}_d\big)\Big) - \widetilde c^*(s,X_{()}(s),V^*(s))\bigg)ds\\
    & \qquad +V^*(s)\widetilde \pi^*(s,X_{()}(s),V^*(s))^\top \widetilde \sigma(s,X_{()}(s))dB(s).
\end{aligned}
\end{equation}
Since $X_{()}$ is an autonomous uncontrolled state process these dynamics can be written in the form
$dV^*(s) = V^*(s)(\alpha(s,V^*(s),\omega)ds + \nu(s,V^*(s),\omega)dB(s))$,
for explicit coefficients $\alpha$ and $\nu$. The Lipschitz and linear growth conditions on $(\widetilde \pi^*,\widetilde c^*)$ of \eqref{eqn:Lipschitz_control} ensure that this SDE for $V^*$ is well-posed and, since $\alpha$ and $\nu$ are $\Fcal^{Y,B}$-predictable with $Y = X_{()}$, the process $V^*$ is $\Fcal^{Y,B}$-adapted. Moreover, $\alpha$ and $\nu$ are bounded so, by Lemma~\ref{lem:Gronwall}, $V^*$ satisfies the bound
\begin{equation} \label{eqn:Lp_V*} 
\E\Big[\sup_{t \leq s \leq T}(V^*(s))^p\Big] \leq C(1+w^p)
\end{equation}
for every $p \geq 1$ and some constant $C>0$. Moreover, from \eqref{eqn:V*} we see by the stochastic exponential formula that $V^*$ satisfies
\[V^*(\cdot) = w\exp\bigg(\int_t^\cdot \big(r + (\widetilde \pi^*)^\top (\widetilde \mu - r{\bf 1}_d) - \frac{1}{V^*}\widetilde c^* - \frac{1}{2}(\widetilde \pi^*)^\top \widetilde a \widetilde \pi^*\big)(s)ds + \int_t^\cdot\widetilde \pi^*(s)^\top \widetilde \sigma(s) dB(s)\bigg),\]
where in the integrands we only display the time-dependence variable $s$ for better readability. Using boundedness of $\widetilde \mu$ and $\widetilde a$ together with the assumed bounds \eqref{eqn:Lipschitz_control} on $\widetilde \pi^*$ and $\widetilde c^*$, which ensure boundedness of $\widetilde \pi^*$ and $\widetilde c^*/V^*$, we obtain $V^*(s) \geq w\exp(-\widetilde C(s-t) +\widehat M(s)) > 0$. Here, $\widetilde C$ is a positive constant and $\widehat M(s) = \int_t^s \widetilde \pi^*(u,X_{()}(u),V^*(u))^\top \widetilde \sigma(u,X_{()}(u))dB(u)$ is a true martingale by boundedness of $\widetilde \pi^*$ and $\widetilde \sigma$.
Next, note that for any $x \in (0,\infty)^d$ and $(\widetilde \pi,c) \in \widetilde A(x_{()},w)$ we have that $(\pi,c) \in A(x,w)$ where $\pi_i = \sum_{k=1}^d\widetilde \pi_k1_{\{\mathcal{R}_i(x) = k\}}$. From this analysis it follows that $(\pi^*,c^*)$ is an admissible control. 

 Now, letting $Z^*(s) = (X_{()}(s),V^*(s))$ we apply It\^o's formula to obtain
\[d\widetilde v(s,Z^*(s)) = -e^{-\beta s} U_1\big(\widetilde c^*(s,Z^*(s))\big)ds + dM^*(s),\] where $M^*$ 
is a local martingale. Here, we used the optimality of $(\pi^*,c^*)$ as in \eqref{eqn:optimal_control_argmax}, the HJB \eqref{eqn:HJB_ranked} to handle the Hamiltonian term, 
and, as before, we used the Neumann boundary condition \eqref{eqn:neumann_condition} to ensure that the reflection terms vanish. Defining the stopping times $\tau_n^* = \inf\{s \geq t: |M^*(s)| \geq n\}$, we see that \eqref{eq:verification_bound} holds with equality when $Z$ and $\tau_n$ are replaced by $Z^*$ and $\tau_n^*$ respectively.
By the polynomial growth condition \eqref{eqn:v_poly_growth} of $\widetilde v$ together with the $L^p$ bounds for $X$ and $V^*$ (guaranteed by Lemma~\ref{lem:Gronwall} and \eqref{eqn:Lp_V*}), we can send $n \to \infty$ and use dominated convergence on the left-hand side together with the terminal condition in \eqref{eqn:HJB_ranked} and monotone convergence on the right-hand side to obtain
\[\E\big[U_2\big(V^*(T)\big) \, \big| \, X(t) = x\big] = \widetilde v(t,x_{()},w) - \E\bigg[\int_t^Te^{-\beta s}U_1\big(\widetilde c^*(s,Z^*(s))\big)ds\, \bigg|\,  X(t) = x\bigg].\] Since $(\pi^*,c^*) \in \Acal_T(t,x,w)$ we deduce that $v(t,x,w) \geq \widetilde v(t,x_{()},w)$, which 
together with \eqref{eqn:upper_bound} shows that $v(t,x,w) = \widetilde v(t,x_{()},w)$ and that $(\pi^*,c^*)$ is an optimal control.
\end{proof}

\section{First-order models} \label{sec:first_order} The HJB equation \eqref{eqn:HJB_ranked} is fully nonlinear and rarely  lends itself to explicit solutions. When coupled with the Neumann boundary condition \eqref{eqn:neumann_condition}, this problem has additional complexity relative to the classical non-ranked consumption-investment problem. Indeed, whenever the asset dynamics are state-dependent, the value function will be too, so that the condition \eqref{eqn:neumann_condition} leads to a different value function than in the situation where $\widetilde \mu$ and $\widetilde \sigma$ directly specify the drift and volatility of the named returns. In this section 
we study the special case of constant coefficients by setting $\widetilde \mu(t,y) = \widetilde \mu$ and $\widetilde \sigma(t,y) = \widetilde \sigma$ for some $\widetilde \mu \in \R^d$ and $\widetilde \sigma \in \R^{d \times d}$, such that $\widetilde a = \widetilde \sigma \widetilde \sigma^\top \in \mathbb{S}^d_{++}$ (the space of symmetric positive definite $d \times d$ matrices), where the problem studied in the previous sections admits an explicit solution. This model for $X$ is known as a \emph{first-order model} and was first proposed in \cite[Chapter~5.5]{fernholz2002stochastic} with a diagonal volatility matrix. As this model has constant drift and volatility coefficients it is reminiscent of the multivariate version of Merton's problem, though first-order models prescribe constant coefficients for the ranked capitalizations rather than their named counterparts. 
In this section we work with the family of power utility functions given for risk-aversion parameter $\gamma \in (0,1)$ and for $w > 0$ by
$U_1(w) = U_2(w) =
\frac{1}{1-\gamma}w^{1-\gamma}$.  The limiting case $\gamma = 1$ corresponds to log utility, and the optimal controls for log utility are recovered by taking $\gamma \to 1$ in the expressions below. Formally, the expressions hold also for $\gamma > 1$, but in this case the utility functions are negative and the hypotheses of Theorem~\ref{thm:verification} are not satisfied. For this reason we restrict our attention in this section to $\gamma \in (0,1)$.
\subsection{Unconstrained problem} \label{sec:unconstrained}
Here we let the constraint set be $\widetilde A(x,w) = \R^d \times [0,\infty)$. In this case the HJB \eqref{eqn:HJB_ranked} becomes precisely the HJB for Merton's problem except on the smaller domain $D_{\geq}$ and with the addition of the Neumann boundary conditions \eqref{eqn:neumann_condition}. It is well-known that the value function for Merton's problem is independent of the asset price variable. Hence, the same solution satisfies \eqref{eqn:HJB_ranked} and \eqref{eqn:neumann_condition}. Indeed, it is straightforward to verify that
\[\widetilde \pi^* = \tfrac{1}{\gamma}\widetilde a^{-1}(\widetilde \mu - r{\bf 1}_d), \qquad \widetilde c^*(t,w) =we^{-\frac{\beta}{\gamma} t}\big /f(t;\nu), \qquad \widetilde v(t,w) = f(t;\nu)^\gamma \tfrac{w^{1-\gamma}}{1-\gamma},\]
where 
\begin{equation} \label{eqn:f}
    f(t;\nu) = e^{\frac{\nu}{\gamma}(T-t)}+\frac{\gamma}{\nu-\beta} e^{-\frac{\nu}{\gamma}t}\Big(e^{\frac{\nu-\beta}{\gamma}T} - e^{\frac{\nu-\beta}{\gamma}t}\Big)
    \ \text{ and } \  \nu = (1-\gamma)\Big(r+\frac{1}
    {2\gamma}(\widetilde \mu - r{\bf 1}_d)^\top \widetilde a^{-1}(\widetilde \mu -r{\bf 1}_d)\Big) 
\end{equation}
satisfy \eqref{eqn:HJB_ranked} and \eqref{eqn:neumann_condition}. In the special case $\nu = \beta$, the correct expression for $f$ is obtained by taking the limit as $\nu \to \beta$ in \eqref{eqn:f}, leading to $f(t;\nu) = e^{\frac{\nu}{\gamma}(T-t)} + (T-t)e^{-\frac{\nu}{\gamma} t}$. As a consequence of Theorem~\ref{thm:verification}, we obtain that the value function is given by $v = \widetilde v$, since the latter is independent of $y$, and the optimal controls are of feedback form, pinned down by the functions
$\pi^*_i(x) = \sum_{k=1}^d \widetilde \pi^*_k1_{\{\mathcal{R}_i(x) = k\}}$ and $c^*(t,w) = \widetilde c^*(t,w).$
Thus, the value function and optimal consumption rate in first-order models are precisely the same as in Merton's problem when the assets have parameters $\widetilde \mu$ and $\widetilde \sigma$. The optimal investment strategy $\pi^*$ also involves the constant Merton fractions $\widetilde \pi^* = \frac{1}{\gamma}\widetilde a^{-1}(\widetilde \mu - r {\bf 1}_d)$ though, unlike in Merton's problem, these fractions are applied to the assets according to their current rank rather than their identity.
\subsection{Open market constraints} \label{sec:open_markets} 
Here we study the case when the investor is only allowed to invest in assets whose rank is in a prespecified nonempty index set $I \subset \{1,\dots,d\}$. Open markets allow the assets available for investment to change over time. This offers a tractable way to incorporate a changing investment universe into the analysis. Moreover, it can serve as a proxy for investors who restrict themselves to certain subsets of the available investment universe. For example, taking $I = \{1,\dots,500\}$ restricts investment to the largest 500 assets in the market, which can serve as a proxy for an investor who chooses to only invest in assets that make up the S\&P 500. We write $I = \{k_1,\dots,k_N\}$, where $k_1 < \dots < k_N$  are ordered indices and $N$ is the number of ranked assets available for investment. 

As in part \ref{item:open_market} of Example~\ref{eg:constraints} and the discussion following it, the open market constraint enforces $\widetilde \pi_k = 0$ whenever $k \not \in I$. As such, making the ansatz that $\widetilde v$ is independent of $y$ again, the Hamiltonian in the HJB \eqref{eqn:HJB_ranked} becomes
\begin{equation} \label{eqn:open_market_Hamiltonian}
    \sup_{\eta \in \R^N}\big\{w\partial_w \widetilde v\big(r+ \eta^\top(\widetilde \mu_I-r{\bf 1}_N)\big) + \tfrac{w^2}{2}\partial_{ww}\widetilde v \eta^\top \widetilde a_I\eta\big\} + \sup_{c \geq 0} \{-c\partial_w\widetilde v +e^{-\beta t}U_1(c)\},
\end{equation}
where $\eta$ corresponds to the non-zero entries of $\widetilde \pi$, while  $\widetilde \mu_I$ and $\widetilde a_I$ are  the truncated vector $(\widetilde \mu_{k_i})_{i =1,\dots,N}$ and matrix $(\widetilde a_{k_ik_j})_{i,j =1,\dots,N}$ respectively. Note that $\widetilde a_I$ is positive definite as it is a principal submatrix of $\widetilde a$. We recognize that \eqref{eqn:open_market_Hamiltonian} corresponds precisely to the Hamiltonian in Merton's problem with $N$ risky assets having drift vector $\widetilde \mu_I$ and diffusion matrix $\widetilde a_I$. As such, in a similar fashion to the unconstrained case, we see that the value function is given by
$v(t,w) = f(t;\nu_I)^\gamma \frac{w^{1-\gamma}}{1-\gamma}$
and the optimal controls are given by the feedback functions 
\begin{align*}
\pi^*_i(x) & = {\textstyle \sum_{\ell=1}^N} \widetilde \eta^*_\ell 1_{\{\mathcal{R}_i(x) = k_\ell\}}, \quad \text{where} \quad \widetilde \eta^* = \frac{1}{\gamma}(\widetilde a_I)^{-1}(\widetilde \mu_I - r{\bf 1}_N),\quad 
 c^*(t,w)  = we^{-\frac{\beta}{\gamma}t} \big / f(t;\nu_I), 
\end{align*}
where
$\nu_I =(1-\gamma)(r+\frac{1}
    {2\gamma}(\widetilde \mu_I - r{\bf 1}_N)^\top (\widetilde a_I)^{-1}(\widetilde \mu_I -r{\bf 1}_N))$ and $f$ is as in \eqref{eqn:f}.
The upshot is that open market constraints, which are intractable in Merton's problem, admit an explicit solution in first-order models. 
Moreover, the rank-based Merton fraction $\widetilde \eta^*$ prescribes the optimal holdings and $(\pi^*,c^*)$ do not depend on $\widetilde \mu_k$ for $k \not \in I$, nor on $\widetilde a_{k\ell}$ unless both  $k,\ell  \in I$.

\subsection{Fully invested constraints} \label{sec:fully_invested}
Here we look at the combined open market and fully invested constraints given by $\widetilde A = \{\widetilde \pi \in \R^d: \widetilde \pi_k = 0 \text{ if } k \not \in I \text{ and } \widetilde \pi^\top {\bf 1}_d = 1\} \times [0,\infty).$ Again, writing $\eta  \in \R^N$  for the nonzero components of $\widetilde \pi$ we obtain the same Hamiltonian \eqref{eqn:open_market_Hamiltonian} with the additional constraint that $\eta^\top {\bf 1}_N = 1$. The problem for $c$ is as before, while the problem  for $\eta$ is a quadratic programming problem with linear equality constraints and can be solved explicitly. Hence, we obtain
\[\eta^* = (\widetilde a_I)^{-1}\bigg(-\frac{\partial_w \widetilde v}{w\partial_{ww}\widetilde v} \widetilde \mu_I + \frac{1+\frac{\partial_w \widetilde v}{w\partial_{ww}\widetilde v}{\bf 1}_N^\top (\widetilde a_I)^{-1}\widetilde \mu_I}{{\bf 1}_N^\top (\widetilde a_I)^{-1}{\bf 1}_N}{\bf 1}_N\bigg), \quad c^* = e^{-\frac{\beta}{\gamma} t}(\partial_w \widetilde v)^{-\frac{1}{\gamma}}. \]
Next, we make the standard ansatz $\widetilde v(t,w) = u(t)\frac{w^{1-\gamma}}{1-\gamma}$ for an unknown function $u$, which leads to the ODE
$0 = u'(t) + \zeta u(t) + \gamma e^{-\frac{\beta}{\gamma} t}u(t)^{1-\frac{1}{\gamma}}$ with terminal condition $u(T) = 1$, where
$\zeta = \frac{1-\gamma}{2\gamma}(\widetilde \mu_I^\top (\widetilde a_I)^{-1}\widetilde \mu_I - \frac{(\gamma - {\bf 1}_N^\top (\widetilde a_I)^{-1}\widetilde \mu_I)^2}{{\bf 1}_N^\top (\widetilde a_I)^{-1}{\bf 1}_N}).$
This ODE has solution $u(t) = f(t;\zeta)^\gamma$, where $f$ is given by \eqref{eqn:f}. As such, the optimal controls are given by the feedback functions
\begin{align*}
    &\pi^*_i(x)  = {\textstyle \sum_{\ell=1}^N}\widetilde \eta^*_\ell1_{\{\Rcal_i(x)=k_\ell\}} \ \text{ with } \ \widetilde \eta^* = \frac{1}{\gamma}(\widetilde a_I)^{-1}\bigg(\widetilde \mu_I + \frac{\gamma-{\bf 1}_N^\top (\widetilde a_I)^{-1}\widetilde \mu_I}{{\bf 1}_N^\top (\widetilde a_I)^{-1}{\bf 1}_N}{\bf 1}_N\bigg), \\
    & c^*(t,w)  = w e^{-\frac{\beta}{\gamma}t}\big / f(t;\zeta).
\end{align*}

\section*{Acknowledgments}
The author is grateful to three anonymous referees whose comments and careful reading improved the quality of the manuscript.

\bibliographystyle{abbrv}
\bibliography{references}

\end{document}